\documentclass[letterpaper, 10 pt, conference]{ieeeconf}  
\IEEEoverridecommandlockouts                        
\usepackage[T1]{fontenc}
\usepackage{cite}

\usepackage{amsmath,amssymb,amsfonts,bm}
\usepackage{amsthm}
\usepackage{graphicx}
\usepackage{textcomp}
\usepackage{xcolor}
\usepackage{color}
\usepackage{hyperref}
\usepackage{amsmath}
\usepackage[labelsep=period]{caption}
\usepackage{graphics} 
\usepackage{epsfig}
\usepackage{times} 
\usepackage{textcomp}
\usepackage{multicol}

\usepackage{float}
\usepackage{dsfont}
\usepackage{anyfontsize} 
\usepackage{multirow}
\usepackage{centernot}
\usepackage{subcaption}
\usepackage{graphicx}
\usepackage[top=54pt, left=54pt, right=54pt, bottom=54pt]{geometry}

\usepackage{enumerate}

\usepackage{algorithm}
\usepackage{algorithmicx}
\usepackage{algpseudocode}
\usepackage{flushend}
\newcommand{\tfm}{\emph{TimesFM}}
\usepackage{textcomp}

\usepackage{mathtools}
\mathtoolsset{showonlyrefs}
\newtheorem{thm}{Theorem}
\newtheorem{problem}{Problem}
\newtheorem{proposition}{Proposition}

\newtheorem{remark}{Remark}

\title{Attack Detection using Time Series Foundation Models}
\author{Sribalaji C. Anand, Anh Tung Nguyen, and George J. Pappas
\thanks{This work is supported by the Swedish Research Council under the grant 2024-00185. Sribalaji C. Anand, and George J. Pappas are with the 
University of Pennsylvania, United States. Sribalaji C. Anand is also affiliated with
KTH Royal Institute of Technology, Sweden. Anh Tung Nguyen is with 
Uppsala University, Sweden. (e-mail: srca@kth.se, pappasg@seas.upenn.edu, anh.tung.nguyen@it.uu.se)}}
\begin{document}
\maketitle
\thispagestyle{empty}
\pagestyle{empty}
%
%
%
\begin{abstract}
This paper addresses the problem of attack detection in cyber-physical systems without any knowledge of the plant model or its structure. A remotely located plant transmits sensor measurements to an operator over a network that is assumed to be under attack. We consider two classes of attacks: model-free replay attacks and model-based stealthy attacks. For the latter, we derive closed-form expressions for the optimal stealthy attack policy against a $\chi^2$ detector, for both linear and nonlinear systems. We then propose a model-structure-free detector based on \emph{TimesFM}, a time-series foundation model developed by Google {Research}, which serves as a surrogate residual generator operating in a zero-shot fashion. We show empirically that the \emph{TimesFM}-based detector achieves a comparable or superior attack detection performance. The efficacy of the proposed approach is demonstrated numerically on the IEEE 14-bus power system. We also demonstrate that \emph{TimesFM} predictions can serve as a substitute for corrupted measurements, a practical mitigation technique when classical redundancy assumptions fail.
\end{abstract}
\section{Introduction}
Attack detection in Cyber-Physical Systems (CPSs) has received considerable research attention in recent years~\cite{giraldo2018survey}. In general, attack detection in CPSs proceeds in two stages: a residual signal is first constructed from the received output signal, and a statistical hypothesis test, such as a $\chi^2$ test and a cumulative-sum test~\cite{murguia2016cusum}, is then applied to the residual to flag anomalies. Both stages, however, rely critically on model knowledge. 
This raises a natural question: \emph{Can we design an efficient attack detector without any explicit knowledge of the system model structure or parameters?}
\setlength{\belowdisplayskip}{2pt}
\setlength{\abovedisplayskip}{2pt}

Driven by seminal works on data-driven control~\cite{van2025data}, there has been a growing body of literature on data-driven attack detection~\cite{krishnan2020data,zhao2022data,shinohara2025detection}. For instance,~\cite{krishnan2020data} derives conditions for the existence of a data-driven perfectly undetectable attack and develops a corresponding data-driven detector. The work~\cite{shinohara2025detection} considers attack detection in the noiseless case where signals are watermarked. However, none of these works address the detection of \emph{optimal} attacks that are designed based on model and detector knowledge. In this work, we aim to design a model-structure-free\footnote{In the sequel, \emph{model-structure-free} refers to the absence of any knowledge of the plant model, neither its structure nor its parameters.} detector for generic nonlinear systems in the presence of noise that can detect both model-free and model-based stealthy attacks (see Fig.~\ref{fig:scheme}). In the spirit of the Swiss Cheese Model of risk management~\cite{reason2006revisiting}, the proposed {model-structure-free} detector is intended to serve as a second layer of defense, deployed alongside existing detectors (possibly already) present in traditional CPSs. To this end, we depart from traditional control-theoretic tools, such as Willems' fundamental lemma~\cite{willems2005note}, and adopt foundation models as surrogate residual generators. In particular, we employ the \tfm\ model developed by Google~\cite{das2024decoder}.

Concretely, we consider a dynamical system located remotely, transmitting sensor measurements over a network for monitoring purposes (see Fig.~\ref{fig:scheme}). An attacker is assumed to tamper with the measurements received by the operator. We consider two attacker types: a \emph{model-free} attacker that conducts replay attacks \cite{falliere2011w32}, and a \emph{model-based} attacker that exploits full system knowledge to construct stealthy attacks \cite{umsonst2022experimental}. The main contributions of this paper are as follows
\begin{enumerate}
    \item \emph{Optimal stealthy attack construction:} We consider the case where the operator employs a model-based $\chi^2$ detector. Assuming the attacker has full model knowledge, we derive closed-form expressions for stealthy attacks that maximally bias the state estimate while remaining stealthy, for both linear and nonlinear systems.
    \item \emph{Model-free detection via foundation models:} We {leverage} \tfm\ as a surrogate residual generator and apply a $\chi^2$ test to its prediction errors (Algorithm~\ref{alg:LTI:detection}). Monte-Carlo empirical results show that model-based stealthy attacks can be reliably detected by \tfm, raising the bar for simultaneous evasion. For model-free replay attacks, we show empirically that \tfm\ achieves comparable or superior detection performance.

    \item \emph{Practical heuristic for overcoming fundamental limitations:} Resilient state estimation  requires detectability under any subset of $N{-}2M$ sensors~\cite{nakahira2018attack}\footnote{Here, attackers compromise any unknown $M$ sensors from $N$ sensors.}. We empirically show that a reliable state estimate can be obtained using \tfm\ even when this assumption fails.
\end{enumerate}

We illustrate the efficacy of our detector on the IEEE 14-bus power system. To the best of our knowledge, this is the first work to propose a foundation-model-based, model-structure-free zero-shot detector against both replay attacks and model-based stealthy attacks, in the presence of noise.

\emph{Related works:} The attack detection literature is vast, and a comprehensive review is beyond the scope of this paper. We instead position our work precisely within this landscape. Broadly, let us consider two players: the attacker and the detector, each of which may be model-based or model-free. Traditionally, a model-based detector can be used to detect model-free replay attacks~\cite{weerakkody2014detecting}. {In the worst-case scenario,} if the attacker has access to the system model, model-based attacks can be designed to be stealthy~\cite{umsonst2022experimental}, {which} we establish formally and illustrate in this paper. By contrast, {thanks to the fact that} \tfm\ has no closed-form innovation structure for the attacker to exploit, attacks designed to be stealthy against the model-based detector are {much harder to be} stealthy against \tfm.

There is also a growing body of work using machine learning architectures for attack detection~\cite{nanduri2016anomaly,kravchik2018detecting}. However, these approaches require task-specific training data, hyperparameter tuning, and retraining whenever the system changes. By contrast, \tfm\ operates in a zero-shot fashion; no training, retraining, or hyperparameter tuning is required, and the same model applies regardless of the system or context/memory length. This makes it a true plug-and-play solution for model-structure-free attack detection.
%

\emph{Notation:} 
Let $A \in \mathbb{R}^{n \times n}$, then $\rho(A)$ denotes its spectral radius. Given $x \in \mathbb{R}, x>0$, $\Gamma(x) = \int_{0}^{\infty} t^{x-1}e^{-t}\text{d}t$ denotes the gamma function. Given $x, y \in \mathbb{R}$ and $y >0$, the Regularised Lower Incomplete Gamma (RLIG) function is given as $P(x,y) = \frac{1}{\Gamma (y)} \int_{0}^x t^{y-1}e^{-t}\text{d}t \triangleq q$, and $P^{-1}(q;y)$ represents the inverse of the RLIG function for a given $y >0$.
\begin{figure}[t]
    \centering
    \includegraphics[width=7.5cm]{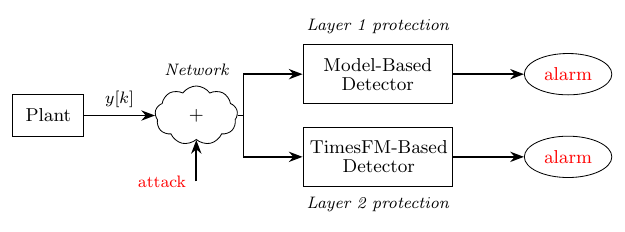}
    \vspace{-5pt}
    \caption{\small A schematic representation of the problem setup.}
    \label{fig:scheme}
    \vspace{-15pt}
\end{figure}
\vspace{-10pt}
\section{Problem Formulation}\label{sec:problem}
In this section, we formalize the CPS setup for studying attack detection, where a remotely located plant periodically transmits its outputs to the plant operator over a network. This communication network is assumed to be under attacks.
\vspace{-5pt}
\subsection{Plant and detector}
Consider a discrete-time dynamical system of the form
\begin{equation}\label{eq:plant}
x[k+1] = f(x[k]) + w[k],\quad y[k] = h(x[k]) + v[k]
\end{equation}
where $x[k] \in \mathbb{R}^n$ is the state of the plant, $y[k] \in \mathbb{R}^m$ is the plant output, and $w[k] \overset{iid}{\sim} \mathcal{N}(0, \Sigma_w)$ and $v[k] \overset{iid}{\sim} \mathcal{N}(0, \Sigma_v)$ are Gaussian process {disturbance} and {Gaussian} measurement noise, respectively. We assume the operator has access to a state estimator of the form
\begin{equation}\label{eq:estimator}
    \begin{bmatrix} \hat{x}_p[k]^\top & \hat{y}_p[k]^\top \end{bmatrix}^\top = \Phi(y[k], \cdot),
\end{equation}
where $\hat{x}_p[k]$ and $\hat{y}_p[k]$ denote the estimated state and output, respectively, and $\Phi$ encodes the plant dynamics (see Remark~\ref{rem:primary}). Using $\Phi$, the operator employs a primary $\chi^2$ detector of the form
\begin{equation}\label{eq:detector}
    g_p[k] > \tau_p \implies \text{alarm}, \; g_p[k] \leq \tau_p \implies \text{no alarm,}
\end{equation}
where $g_p[k] \triangleq z_p[k]^\top \Sigma_p^{-1} z_p[k]$ is the $\chi^2$ detection statistic, $z_p[k] \triangleq y[k]-\hat{y}_p[k]$ is the residue signal, $\Sigma_p$ is the residue covariance, and $\tau_p$ is chosen to achieve a desired nominal FAR $\alpha_p$. As discussed in the introduction, model-based attacks can be designed to be stealthy against $g_p$ by construction. The operator therefore also deploys a secondary, {model-structure-free} detector $g_s$ based solely on historically observed data
\begin{equation}
    g_s(\cdot,k) > \tau_s \implies \text{alarm}, \; g_s(\cdot,k) \leq \tau_s \implies \text{no alarm,}
\end{equation}
where $\tau_s$ is chosen to achieve a desired nominal FAR $\alpha_s$. The design of both $\tau_p$ and $\tau_s$ is discussed in detail later.
\begin{remark}\label{rem:primary}
The primary estimator and detector in~\eqref{eq:estimator} and~\eqref{eq:detector} are assumed to characterize the optimal stealthy attacks in Theorems~\ref{thm:impact:LTI} and~\ref{thm:impact:NL}. The TimesFM-based detector operates independently of $g_p$; attacks designed to be stealthy against $g_p$ are, in practice, detectable by monitoring $g_s$. 
$\hfill\triangleleft$
\end{remark}
\vspace{-10pt}
\subsection{Attacker}\label{sec:adversary}
We consider an attacker that corrupts the sensor measurements transmitted to the operator as follows
\begin{equation}
    \tilde{y}_i[k] \triangleq \varphi_i(y_i[k]), \quad \forall\, i \in \{1,\dots,m\},\ \forall\, k \geq k_a, \label{def_attack_output}
\end{equation}
where $\varphi_i: \mathbb{R} \to \mathbb{R}$ and $k_a$ is the time at which the attack begins. We consider two classes of attacks: model-free replay attacks and model-based stealthy attacks. 
\subsubsection{Model-free replay attacks}
Replay attacks are of particular practical relevance, as real-world instances, such as the Stuxnet malware~\cite{falliere2011w32}, have demonstrated their potential to cause catastrophic consequences in CPSs. Thus, let us consider an attacker that records sensor data during the time interval $[\kappa_0, \kappa_1]$, with $T = \kappa_1 - \kappa_0 + 1$. Then, the replay attack takes the attack policy \eqref{def_attack_output} in the following form
\begin{equation}
\varphi_i(y_i[k]) \triangleq y_i\bigl[\kappa_0 + (k - k_a) \bmod T\bigr], \; \forall\, i \in \{1,\dots,m\},
\end{equation}
for all $k \geq k_a$, where $k_a \geq \kappa_1$ {and $\bmod$ stands for the modulus operator}. The attacker loops through the recorded sequence $\{y_i[\kappa_0], \dots, y_i[\kappa_1]\}$, possibly indefinitely, so the operator continuously receives a repeating window of legitimate-looking data while the true plant state evolves freely. Such attacks require no model knowledge and exploit the fact that the system operates in steady state.
\subsubsection{Model-based maximum-impact stealthy attacks}
From a risk management perspective, one must consider the worst-case scenario in which the attacker has exact model knowledge. In this case, attacks that maximally degrade the operator's state estimate while remaining undetected by $g_p$ can be designed. Taking the attack policy~\eqref{def_attack_output} in the additive form $\varphi_i(y_i[k]) \triangleq y_i[k] + a_i[k]$, the attack design problem becomes
\begin{equation}\label{opt:stealthy}
\scalebox{0.91}{$\underset{a}{\sup} \left\{\mathbb{E}\left[w^\top \epsilon_p[T]\right] \big\vert \; \mathbb{P}\left(g_p^a[k] > \tau_p\right) \leq \alpha_p + \Delta\alpha_p, \forall\, k \geq k_a\right\}$}
\end{equation}
where $w \in \mathbb{R}^n$ is a weighting vector, $\epsilon_p[T] \triangleq \hat{x}_p^a[T] - \hat{x}_p[T]$ is the attack-induced estimation deviation at the time $T$, $\hat{x}_p^a[k]$ is the state estimate under attack, $g_p^a[k]$ is the $\chi^2$ detection statistic under attack, and $\Delta\alpha_p \geq 0$ is the allowed increase in false alarm rate. In practice, $\tau_p$ is chosen to yield the nominal FAR $\alpha_p$ asymptotically; over any finite window, the empirical FAR may deviate from $\alpha_p$, making a small increase $\Delta\alpha_p$ indistinguishable from natural fluctuations~\cite{milovsevivc2017analysis}. We next provide closed-form solutions to~\eqref{opt:stealthy}.
\begin{remark}
The design problem~\eqref{opt:stealthy} is formulated with respect to the primary detector $g_p$; the attacker is assumed to have no knowledge of the secondary detector $g_s$. This asymmetry in attacker knowledge is analogous to the role of watermarking schemes~\cite{weerakkody2014detecting} that enable attack detection. $\hfill \triangleleft$
\end{remark}

\vspace{-10pt}
\subsection{Problem statement}
Before formally stating the problem, we briefly introduce the \tfm\ foundation model used throughout this paper. \textit{TimesFM}~\cite{das2024decoder} is a time-series foundation model developed by Google, trained on billions of real-world time-series data points spanning diverse domains. Given a history of observations, \textit{TimesFM} predicts the next values of a time series along with an associated variance, enabling uncertainty-aware forecasting. Crucially, the model operates in a zero-shot fashion, meaning \textit{no retraining or fine-tuning on plant data} is required. However, the model is univariate in the sense that each output channel $y_i[k]$ is forecast independently, without explicitly modelling inter-channel correlations. We are now ready to formally state the problem studied in this paper.
\begin{problem}
Can time-series foundation models be used to realize the secondary detector $g_s$, so as to detect attacks on dynamical systems without any model knowledge, including attacks that are specifically designed to be stealthy against the primary model-based detector $g_p$? $\hfill \triangleleft$
\end{problem}
\vspace{-10pt}
\section{Optimal Model-based Attacks}\label{sec:attacks}
This section is organized as follows. In Proposition~\ref{prop:tau:MB}, we adopt a result from~\cite{murguia2016cusum} to derive the detector threshold $\tau_p$ for the $\chi^2$ detector given a desired nominal FAR. In Theorem~\ref{thm:impact:LTI}, we derive the closed-form solution to the optimization problem~\eqref{opt:stealthy} for LTI systems. Finally, in Theorem~\ref{thm:impact:NL}, we derive the solution to~\eqref{opt:stealthy} for nonlinear systems.

\begin{proposition}\label{prop:tau:MB}
Suppose when there are no attacks, it holds that $z_p[k] \overset{iid}{\sim} \mathcal{N}(0, \Sigma_p)$. Then the threshold $\tau_p$ yielding a desired FAR $\alpha_p$ is given by $\tau_p = 2P^{-1}\!\left(\frac{m}{2},\ 1 - \alpha_p\right)$, where $P^{-1}(\cdot,\cdot)$ is the inverse of the RLIG function.$\hfill \square$
\end{proposition}
%
\begin{remark}
The residual decomposes as $z_p[k] = v[k] + \varepsilon[k]$, where $\varepsilon[k] \triangleq h(x[k]) - \hat{y}_p[k]$ is the predictor error and $v[k]$ is the measurement noise in~\eqref{eq:plant}. In the absence of process noise ($w[k] = 0$), the state $x[k]$ evolves deterministically, so $\varepsilon[k]$ is a function of past outputs and is therefore independent of the current measurement noise $v[k]$. Under this independence, if $\varepsilon[k]$ is approximately Gaussian with covariance $\Sigma_\varepsilon$, then $z_p[k] \sim \mathcal{N}(0, \Sigma_v + \Sigma_\varepsilon)$, justifying the assumption in Proposition~\ref{prop:tau:MB} with $\Sigma_p \triangleq \Sigma_v + \Sigma_\varepsilon$. $\hfill\triangleleft$
\end{remark}

Proposition~\ref{prop:tau:MB} provides the threshold design for the $\chi^2$ detector under Gaussianity of the detection statistic, independently of any assumption on the plant dynamics. We next use this threshold to derive the optimal attack policy.
\vspace{-5pt}
\subsection{Optimal attacks against LTI systems}
We now solve the optimization problem~\eqref{opt:stealthy} for the case where the plant is LTI of the form $f(x[k]) = Ax[k]$, $h(x[k])=Cx[k]$ where $\rho(A) \leq 1$. To this end, we consider that $\Phi(\cdot)$ in~\eqref{eq:estimator} 
the following form
\begin{equation}\label{eq:observer}
    \begin{aligned}
        \hat{x}_p[k+1] &= A\hat{x}_p[k] + Kz_p[k],\\
        z_p[k] &= \tilde{y}[k] - C\hat{x}_p[k],\;\;\hat{y}_p[k] = C\hat{x}_p[k],
    \end{aligned}
\end{equation}
where $z_p[k]$ is the residual signal and $K$ is the observer gain. The gain $K$ may represent a fixed Luenberger gain or the steady-state Kalman gain, the latter being applicable when $\Sigma_w \succ 0$ and the Kalman gain has converged to its steady-state value. {Hence, it holds that $\rho(A - KC) < 1$. The following theorem presents a closed-form solution to \eqref{opt:stealthy}.}
\begin{thm}\label{thm:impact:LTI}
Consider the system~\eqref{eq:plant} in its LTI form and the corresponding observer~\eqref{eq:observer}. Suppose that $z_p[k] \overset{iid}{\sim} \mathcal{N}(0, \Sigma_p)$, and define the attack budget
\begin{equation}\label{eq:attack:budget}
    \Delta\tau_p^* \triangleq \sup\left\{\lambda \geq 0 : 
    1 - F_{\chi^2(m,\lambda)}(\tau_p) \leq \alpha_p + \Delta\alpha_p
    \right\},
\end{equation}
where $F_{\chi^2(m,\lambda)}$ is the cumulative distribution function of the non-central $\chi^2$ distribution. Then the solution to~\eqref{opt:stealthy} is given by $a^\star[k] = C\epsilon_p[k] + \delta^\star[k]$, where $\delta^\star[k] = \sqrt{\Delta\tau_p^*}\, \frac{\Sigma_p c[k]}{\sqrt{c[k]^\top \Sigma_p c[k]}}$, and $c[k] \triangleq K^\top\!\left(A^{T-1-k}\right)^\top w.$
\end{thm}
\begin{proof}
Under attack, $z_p^a[k] = \tilde{y}[k] - C\hat{x}_p^a[k]$. By basic algebra, it follows that $z_p^a[k] = z_p[k] + \delta[k]$, where $\delta[k] \triangleq a[k] - C(\hat{x}_p^a[k] - \hat{x}_p[k])$. Subtracting the nominal observer from the attacked observer dynamics gives $\epsilon_p[k+1] = A\epsilon_p[k] + K\delta[k]$, $\epsilon_p[k_a] = 0$, which unrolls as $\epsilon_p[T] = \sum_{k=k_a}^{T-1}A^{T-1-k}K\delta[k]$. Since $\delta[k]$ is the attacker's design variable, we restrict attention to deterministic $\delta[k]$. Taking expectation on both sides, and using linearity of expectation, it follows that $ \mathbb{E}[w^\top\epsilon_p[T]] = \textstyle \sum_{k=k_a}^{T-1} c[k]^\top\delta[k]$, where $c[k]$ is given in the theorem statement.

Since $z_p[k] \sim \mathcal{N}(0, \Sigma_p)$ and $\delta[k]$ is deterministic, $z_p^a[k] = z_p[k] + \delta[k] \sim \mathcal{N}(\delta[k], \Sigma_p)$. Define $\tilde{z}_p^a[k] \triangleq \Sigma_p^{-1/2}z_p^a[k] \sim \mathcal{N}(\Sigma_p^{-1/2}\delta[k], I_m)$. Then $ g_p^a[k] = \|\tilde{z}_p^a[k]\|^2 $ which is the sum of squares of $m$ independent Gaussian random variables with unit variance and means $(\Sigma_p^{-1/2}\delta[k])_i$. By definition of the non-central $\chi^2$ distribution
\begin{equation}
g_p^a[k] \sim \chi^2(m,\, \lambda), \; \lambda \triangleq \|\Sigma_p^{-1/2}\delta[k]\|^2 = \delta[k]^\top\Sigma_p^{-1}\delta[k].
\end{equation}
Therefore $P(g_p^a[k] > \tau_p) = 1 - F_{\chi^2(m,\lambda)}(\tau_p)$. Since $1 - F_{\chi^2(m,\lambda)}(\tau_p)$ is strictly increasing in $\lambda$, the constraint $P(g_p^a[k] > \tau_p) \leq \alpha_p + \Delta\alpha_p$ is equivalent to $\lambda \leq \Delta\tau_p^*$, i.e.,
$\delta[k]^\top\Sigma_p^{-1}\delta[k] \leq \Delta\tau_p^*$, which decouples across time steps, then problem~\eqref{opt:stealthy} also decouples as
\begin{equation}
    \max_{\delta[k]}\ c[k]^\top\delta[k] \quad \text{s.t.} \quad 
    \delta[k]^\top\Sigma_p^{-1}\delta[k] \leq \Delta\tau_p^*.
\end{equation}
Applying the Cauchy--Schwarz inequality in the $\Sigma_p^{-1}$-weighted inner product gives
%
\begin{equation}
\scalebox{0.8}{$c[k]^\top\delta[k] \leq \sqrt{c[k]^\top\Sigma_p c[k]}\cdot \sqrt{\delta[k]^\top\Sigma_p^{-1}\delta[k]} \leq \sqrt{\Delta\tau_p^*}\sqrt{c[k]^\top\Sigma_p c[k]}$}.
\end{equation}
Equality holds when $\Sigma_p^{-1/2}\delta[k] = \frac{\sqrt{\Delta\tau_p^*}\,\Sigma_p^{1/2}c[k]}{\sqrt{c[k]^\top\Sigma_p c[k]}}$, giving $\delta^\star[k]$ in the theorem statement, completing the proof.
\end{proof}

Theorem~\ref{thm:impact:LTI} provides a closed-form stealthy attack policy for the case where the operator employs a static observer. We now extend this result to the nonlinear setting, where the operator employs an Extended Kalman Filter (EKF) and the attack design must account for the time-varying gain.
\vspace{-10pt}
\subsection{Optimal attacks against nonlinear systems}
We now address~\eqref{opt:stealthy} for the case where the plant is nonlinear. To this end, we consider that $\Phi(\cdot)$ in~\eqref{eq:estimator} is an EKF in the numerically stable Joseph form as follows
\begin{equation}\label{eq:ekf}
    \begin{aligned}
        \hat{x}_p[k|k-1] &= f(\hat{x}_p[k-1|k-1]),\\
        P_{k|k-1} &= F_kP_{k-1|k-1}F_k^\top + \Sigma_w,\\
        S_k &= H_kP_{k|k-1}H_k^\top + \Sigma_v,\\
        K_k &= P_{k|k-1}H_k^\top S_k^{-1},\\
        z_p[k] &= \tilde{y}[k] - h(\hat{x}_p[k|k-1]),
    \end{aligned}
\end{equation}
\begin{equation}
        \begin{aligned}
                \hat{x}_p[k|k] &= \hat{x}_p[k|k-1] + K_k z_p[k],\\
        P_{k|k} &= G_kP_{k|k-1}G_k^\top + K_k\Sigma_vK_k^\top,
    \end{aligned}
\end{equation}
where $G_k = I - K_kH_k$, $F_k = \partial f/\partial x|_{\hat{x}_p[k-1|k-1]}$ and $H_k = \partial h/\partial x|_{\hat{x}_p[k|k-1]}$ are the Jacobians of $f$ and $h$ evaluated at the current estimates, $P_{k|k-1}$ is the predicted error covariance, $S_k$ is the innovation covariance, and $K_k$ is the time-varying EKF gain. Let $\hat{x}_p^a[k|k]$ denote the EKF state estimate under attack, and define the update deviation $\epsilon_p[k|k] \triangleq \hat{x}_p^a[k|k] - \hat{x}_p[k|k]$. We are now ready to state the main result.
\begin{thm}\label{thm:impact:NL}
Consider the system~\eqref{eq:plant} and the corresponding EKF~\eqref{eq:ekf}. Suppose that the EKF has converged and that the innovation sequence satisfies $z_p[k] \approx \mathcal{N}(0, S_k)$, where $S_k$ is the innovation covariance from~\eqref{eq:ekf}. Consider the attack budget in \eqref{eq:attack:budget}, and suppose that the attack-induced deviation $\epsilon_p[k|k]$ remains small enough. Then an approximate solution to~\eqref{opt:stealthy} is given by $a^\star[k] = h(\hat{x}_p[k|k-1] + \epsilon_p[k|k-1]) - h(\hat{x}_p[k|k-1]) + \delta^\star[k]$, where $\delta^\star[k] = \sqrt{\Delta\tau^*}\, \frac{S_k c[k]}{\sqrt{c[k]^\top S_k c[k]}}$, with $c[k]$ computed via the backward recursion $\lambda[T] = w$,
\begin{equation}
\scalebox{0.9}{$c[k] = K_k^\top \lambda[k+1], \; \lambda[k] = F_k^\top \lambda[k+1], \; k = T-1, \dots, k_a.$}
\end{equation}
\end{thm}
\begin{proof}
Under attack, adding and subtracting $h(\hat{x}_p[k|k-1])$ from the attacked innovation gives $z_p^a[k] = z_p[k] + \delta[k]$, where $\delta[k] \triangleq a[k] - (h(\hat{x}_p[k|k-1] + \epsilon_p[k|k-1]) - h(\hat{x}_p[k|k-1]))$. Applying a first-order linearization of the deviation dynamics around the nominal trajectory gives $\epsilon_p[k+1|k] \approx F_k\epsilon_p[k|k-1] + K_k\delta[k]$, and introducing the backward recursion $\lambda[T] = w$, $\lambda[k] = F_k^\top\lambda[k+1]$, $c[k] = K_k^\top\lambda[k+1]$ gives $\mathbb{E}[w^\top\epsilon_p[T|T-1]] \approx \sum_{k=k_a}^{T-1} c[k]^\top\delta[k]$. Since $z_p[k] \approx \mathcal{N}(0, S_k)$, the stealthiness constraint $P(g_p^a[k] > \tau_p) \leq \alpha_p + \Delta\alpha_p$ is equivalent to $\delta[k]^\top S_k^{-1}\delta[k] \leq \Delta\tau^*$ by the same argument as in the proof of Theorem~\ref{thm:impact:LTI}, with $\Sigma_p$ replaced by $S_k$. The remainder of the proof follows identically {and is omitted due to space limitations.} 
\end{proof}

Theorem~\ref{thm:impact:NL} provides a closed-form attack policy for nonlinear systems. We next show that such attacks are in practice stealthy against model-based detectors, and we introduce a \tfm-based algorithm to enhance detection.
\begin{remark}
While the small-deviation assumption in Theorem~\ref{thm:impact:NL} is difficult to verify analytically, the numerical results on the IEEE 14-bus system confirm that the constructed attack is consistently stealthy against the EKF-based detector across all trials, validating the approximation in practice. $\hfill\triangleleft$
\end{remark}
\vspace{-10pt}
\section{\tfm-based Attack Detection}\label{sec:detection}
In this section, we present a data-driven attack detection algorithm using \tfm\ to realize the secondary detector $g_s$. The algorithm is inspired by the classical $\chi^2$ detector, but replaces the model-based residuals with prediction residuals generated by \tfm. As summarized in Algorithm~\ref{alg:LTI:detection}, the approach proceeds in three phases. In the first phase, a warmup period allows \tfm\ to observe sufficient data to form meaningful predictions. Specifically, the oracle \tfm\ takes a window of $L$ past sensor outputs and predicts the next output $\hat{y}$ (line~$5$ in Algorithm~\ref{alg:LTI:detection}). In the second phase, a clean window of sensor data is used to estimate the covariance of the \tfm\ prediction residuals under nominal operation. In the third phase, this estimated covariance is used to construct a $\chi^2$-type test statistic for attack detection. Crucially, the context buffer is only updated when no alarm is raised, preventing corrupted measurements from contaminating future predictions (line $12$ in Algorithm~\ref{alg:LTI:detection}). The algorithm applies to both linear and nonlinear systems, requiring no model knowledge. We next depict the efficacy of Algorithm~\ref{alg:LTI:detection} using numerical experiments. The code to reproduce all the results can be found in \url{https://github.com/balajianand1994/Attack_detection_using_TFM.git}
\setlength{\floatsep}{4pt}
\setlength{\textfloatsep}{4pt}
\begin{algorithm}[t]
\caption{Attack Detection using \textit{TimesFM}}\label{alg:LTI:detection}
\begin{algorithmic}[1]
\Require Sensor outputs $\tilde{y}_i[k]$, context length $L$, 
         warmup length $T_w$, clean window length $T_c$, FAR $\alpha_s$
\Ensure Alarm signal at each time step $k$
\For{$k = 1, \ldots, T_w$} \Comment{Warmup}
    \State Store $\tilde{y}[k]$ in buffer $\mathcal{B}$
\EndFor
\For{$k = T_w+1, \ldots, T_w+T_c$} 
    \State $r[k] \leftarrow \tilde{y}[k] - \texttt{TimesFM}\big(\mathcal{B}[k{-}L:k{-}1]\big)$; update $\mathcal{B}$
\EndFor
\State $\hat{\Sigma} \leftarrow \frac{1}{T_c-1}\sum_{k=T_w+1}^{T_w+T_c} r[k]r[k]^\top$; 
       set $\tau_s$ via Proposition~\ref{prop:tau:MB}
\For{$k = T_w+T_c+1, \ldots$} \Comment{Online detection}
    \State $\hat{y}[k] \leftarrow \texttt{TimesFM}\big(\mathcal{B}[k{-}L:k{-}1]\big)$
    \State $g_s[k] \leftarrow (\tilde{y}[k]{-}\hat{y}[k])^\top \hat{\Sigma}^{-1} (\tilde{y}[k]{-}\hat{y}[k])$
    \If{$g_s[k] > \tau_s$} \textbf{raise alarm}
    \Else\ update $\mathcal{B}$ with $\theta\tilde{y}[k] + (1{-}\theta)\hat{y}[k]$
    \EndIf
\EndFor
\end{algorithmic}
\end{algorithm}

\vspace{-5pt}
\subsection{Numerical results: Linear case}
We consider an undamped mass-spring system of the form~\eqref{eq:plant} with $f(x[k]) = Ax[k]$ and $h(x[k]) = Cx[k]$, where
\begin{equation}
    A = \begin{bmatrix} \cos(\omega \Delta t) & \frac{\sin(\omega \Delta t)}{\omega} \\ -\omega\sin(\omega \Delta t) & \cos(\omega \Delta t) \end{bmatrix}, \; C = \begin{bmatrix} 1 & 0 \\ 0.31 & -0.48\\ -0.21 & 0.43 \end{bmatrix},
\end{equation}
with $\omega = 0.3$, $\Delta t = 1$, and $x[0] = \begin{bmatrix} 1 & 0 \end{bmatrix}^\top$. Here, the state $x[k] \in \mathbb{R}^2$ represents the displacement and velocity of the mass, and $y[k] \in \mathbb{R}^3$ represents three linear sensor measurements of the state. We assume no process noise and Gaussian measurement noise with standard deviation $\sigma = 0.01$. For the primary detector, we design a Luenberger observer with gain $K$ placing the closed-loop poles of $A - KC$ at $0.5 \pm 0.1j$, with threshold $\tau_p = 12.838$ corresponding to a nominal FAR of $0.5\%$. For the secondary detector, we use Algorithm~\ref{alg:LTI:detection} with $L = 50$, $T_w = 51$, $T_c = 20$, and $\theta = 0.8$.
\subsubsection{Model-free replay attacks}
We consider a replay attack on all sensors simultaneously, which is particularly challenging to detect since no sensor provides a clean reference. The attacker records data over $[\kappa_0, \kappa_1] = [50, 70]$ and begins replaying at $k_a = 71$. The detection results and the attacked sensor trajectory are shown in Fig.~\ref{fig:LTI:results}. During nominal operation, \tfm\ achieves a FAR of $0.5\%$, compared to $2.8\%$ for the observer-based detector. Under attack, \tfm\ consistently raises alarms from $k = 91$ onward, while the observer-based detector raises alarms only over $k \in [91, 93]$, which can be dismissed as false alarms given its higher nominal FAR. Overall, \tfm\ achieves both a lower FAR and more sustained attack detection.
%
%
\subsubsection{Model-based stealthy attacks}
We consider a stealthy attack constructed using Theorem~\ref{thm:impact:LTI} with $\Delta\tau_p = 0.385$, starting at $k_a = 81$ and running until $k_2 = 100$. The results are shown in Fig.~\ref{fig:LTI:results}. During nominal operation, the FAR of \tfm\ remains consistent with the replay-attack scenario. Under attack, the EKF-based detector raises alarms indistinguishably from its nominal FAR, as expected by construction. In contrast, the \tfm-based detector successfully detects the attack, demonstrating that attacks designed to be stealthy against model-based detectors are not stealthy against \tfm.
%
%
\begin{figure}[t]
    \centering
    \includegraphics[width=8cm]{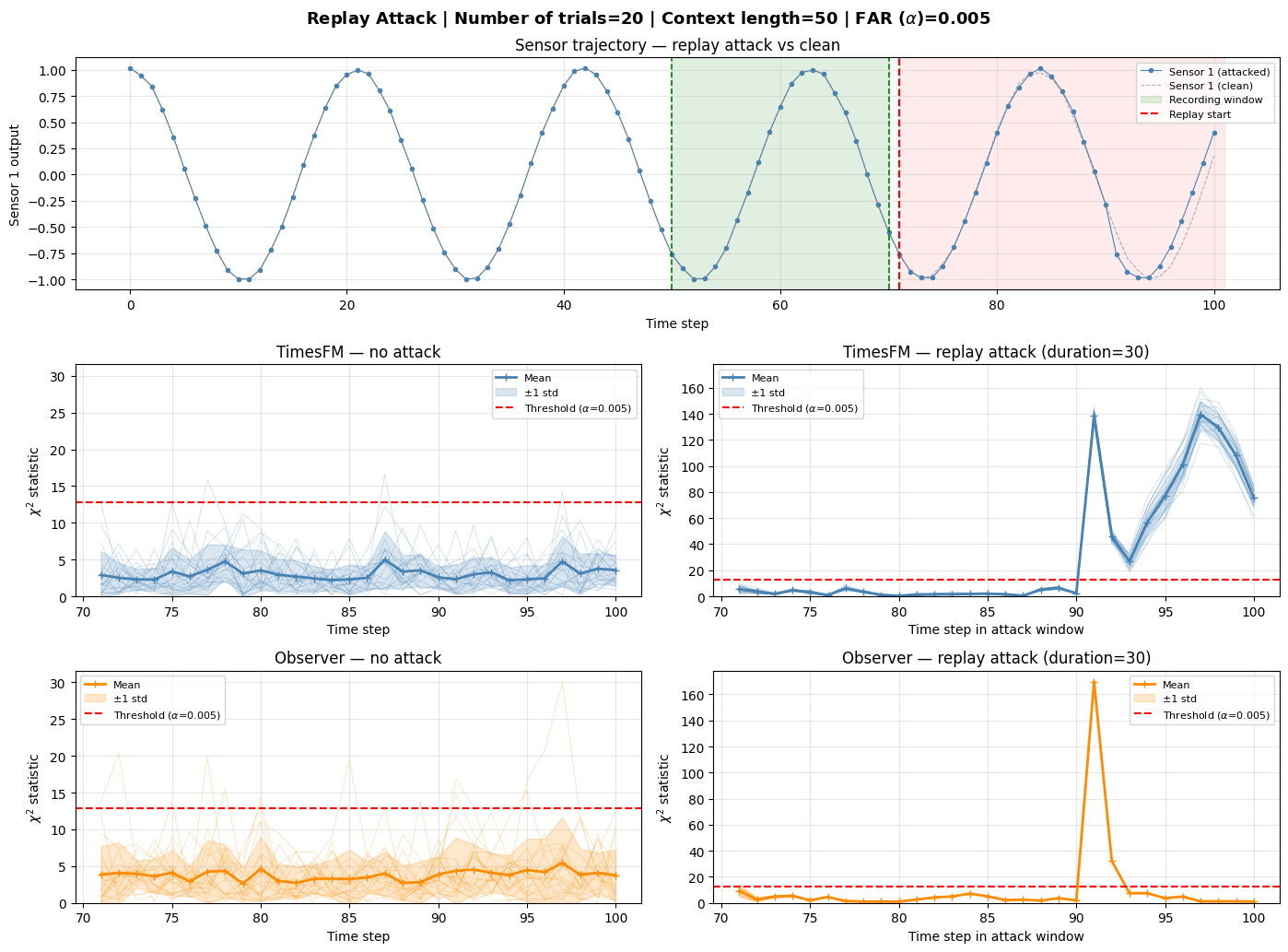}\\
    \vspace{4pt}
    \includegraphics[width=8cm]{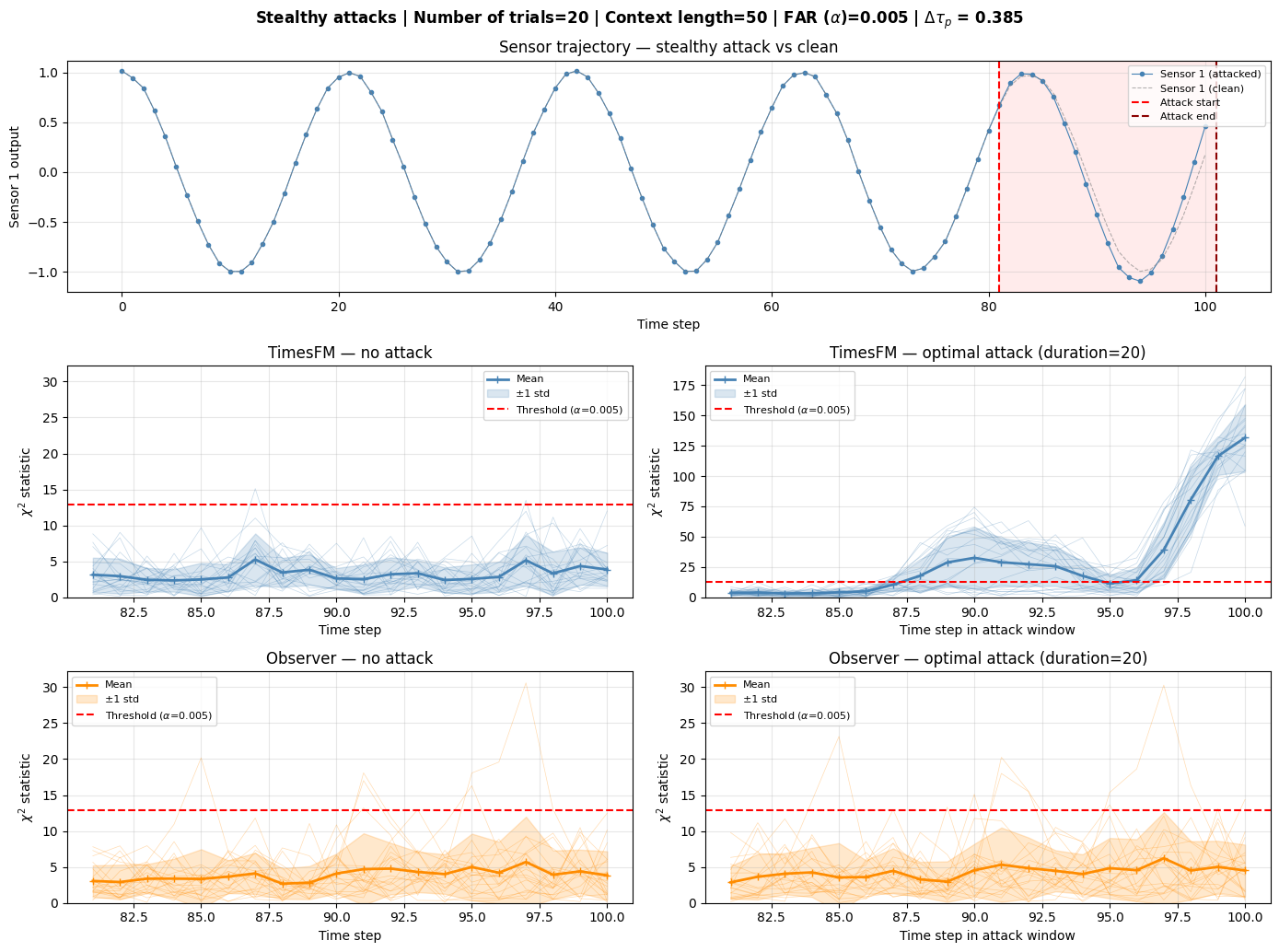}
    \caption{\small Evaluation of Algorithm~\ref{alg:LTI:detection} under replay attacks (top three rows) and stealthy attacks constructed using Theorem~\ref{thm:impact:LTI} (bottom three rows). For each attack type: (Top) sensor output $y_1[k]$ during nominal operation and attack phase; (Middle) mean and standard deviation of the test statistic across $20$ trials alongside threshold $\tau_p$, for the \tfm-based detector under nominal operation (left) and under attack (right); (Bottom) same for the observer-based detector.}
    \label{fig:LTI:results}
\end{figure}
\subsubsection{Robust state estimation}
We now consider a single-sensor setting with $C = \begin{bmatrix} 1 & 0 \end{bmatrix}$, where the only available sensor is under attack.
{Note that existing model-based secure state estimation methods do not apply due to the fundamental limitation~\cite{nakahira2018attack} as they require sufficient observability redundancy.}
Instead, we suggest a different mitigation approach: once an attack is detected by Algorithm~\ref{alg:LTI:detection}, we replace the corrupted measurement with the \tfm\ point prediction and feed this into the observer. The results are shown in Fig.~\ref{fig:mitigation}. Even when the only sensor is compromised, the \tfm-protected observer consistently achieves a lower state reconstruction error than the unprotected observer, demonstrating the practical value of \tfm\ as a mitigation tool beyond detection alone.
\begin{figure}[tb]
    \centering
    \includegraphics[width=6cm]{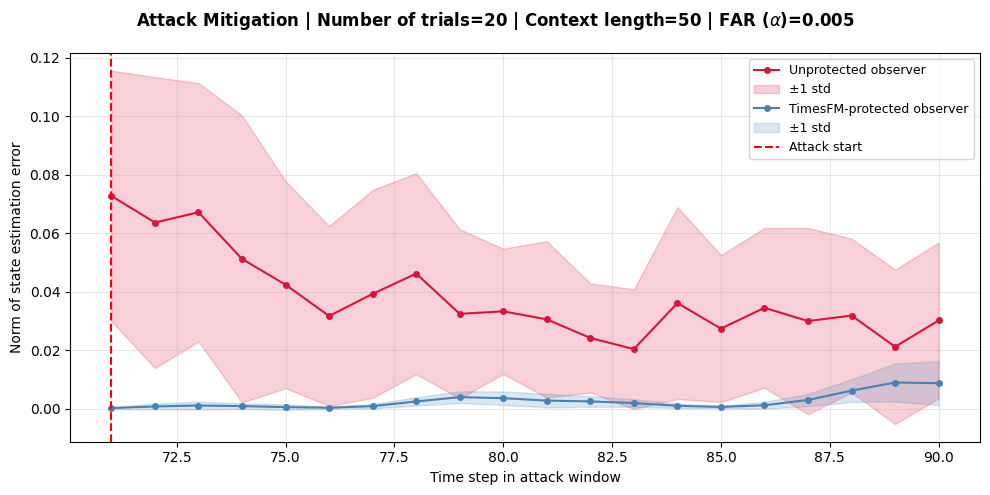}
    \caption{\small Mean and standard deviation of the state reconstruction error across $20$ trials for the unprotected observer (red) and the \tfm-protected observer (blue), where corrupted measurements are replaced by \tfm\ point predictions upon detection.}
    \label{fig:mitigation}
\end{figure}
\subsubsection{Robustness against an adaptive attacker}
\label{sec:adaptive}
We now consider an attacker with black-box query access to $g_s$, who simultaneously attempts to evade both detectors. We design attacks via SPSA~\cite{s2013stochastic}, warm-started using the attack in Theorem~1 and refined over $20$ iterations, scaled to satisfy the empirical FAR constraint on $g_s$. We consider the single-sensor setting $C = 
\begin{bmatrix}1 & 0\end{bmatrix}$ over $N=10$ trials. The oblivious attacker of Theorem~1, designed against 
$g_p$ alone, achieves a state deviation of $0.637$. Under simultaneous evasion, the mean state deviation is $0.007$, a reduction of $91\times$, confirming that $g_s$ severely limits the adaptive attacker's impact even under query access.

\vspace{-5pt}
\subsection{Numerical results: IEEE 14-bus power system}
We next demonstrate the efficacy of Algorithm~\ref{alg:LTI:detection} on the nonlinear IEEE 14-bus power system, which consists of $14$ buses and $20$ transmission lines. The dynamics of bus $i \in \{1,\dots,14\}$ are governed by the swing equation
\begin{equation}\label{bus_dyn}
    m_i \ddot{p}_i + h_i \dot{p}_i = \textstyle -\sum_{j \in \mathcal{N}_i} B_{ij} \sin(p_i - p_j),
\end{equation}
where $m_i > 0$ and $h_i \geq 0$ are the inertia and damping coefficients of bus $i$, $B_{ij} > 0$ is the susceptance of the line connecting buses $i$ and $j$, and $p_i$ denotes the phase angle of bus $i$. Bus~$1$ is the reference bus with $p_1 \equiv 0$, and the system is discretized with step $\Delta t = 0.02$\,s via forward Euler. The measurement vector $y[k] \in \mathbb{R}^{34}$ consists of all $14$ bus frequencies and all $20$ line active power flows, with measurement noise covariance $\Sigma_v = \mathrm{diag}(\sigma_\omega^2 \mathbf{1}_{14}, \sigma_f^2 \mathbf{1}_{20})$, where $\sigma_\omega = 0.035$ and $\sigma_f = 0.050$. For the primary detector, we design an EKF in~\eqref{eq:ekf} with  $\alpha_p = 0.1\%$. For Algorithm~\ref{alg:LTI:detection} we use $L = 100$, $T_w = 110$, $\theta = 1$, and $T_c = 200$.
\begin{remark}
The threshold $\tau_s$ in Algorithm~\ref{alg:LTI:detection} is designed  assuming approximately Gaussian residuals under nominal operation. A Shapiro-Wilk test on the nominal residuals confirms this: $32$ ($31$) out of $34$ channels pass at the $5\%$ significance level for the EKF (\tfm)-based detector, with $W \in [0.990, 0.998]$ across all channels, indicating 
only mild departures from Gaussianity. $\hfill\triangleleft$
\end{remark}
\subsubsection{Model-free replay attacks}
We consider a replay attack applied simultaneously to all sensors. The attacker records data over $[\kappa_0, \kappa_1] = [60, 70]$ and begins replaying at $k_a = 71$. The detection results and the attacked sensor trajectory are shown in Fig.~\ref{fig:NL:results}. Under attack, \tfm\ consistently raises alarms throughout the attack period. The EKF-based detector raises alarms only during the initial phase of the attack, after which the replayed data becomes indistinguishable from nominal behaviour in the model-based residual. We note that the nominal FAR of \tfm\ is slightly elevated in this experiment due to the limited covariance estimation horizon; a longer clean window $T_c$ is expected to reduce it.
\subsubsection{Model-based stealthy attacks}
We consider a stealthy attack constructed using Theorem~\ref{thm:impact:NL} with $\Delta\tau_p = 5.87$, starting at $k_a = 111$ and running until $k_2 = 161$. The results are shown in Fig.~\ref{fig:NL:results}. Under attack, the EKF-based detector raises no alarms, as expected. In contrast, the \tfm-based detector consistently detects the attack across all trials, demonstrating that attacks designed to be stealthy against model-based detectors are not stealthy against \tfm, even for nonlinear systems.
\begin{figure}[tb]
    \centering
    \includegraphics[width=8cm]{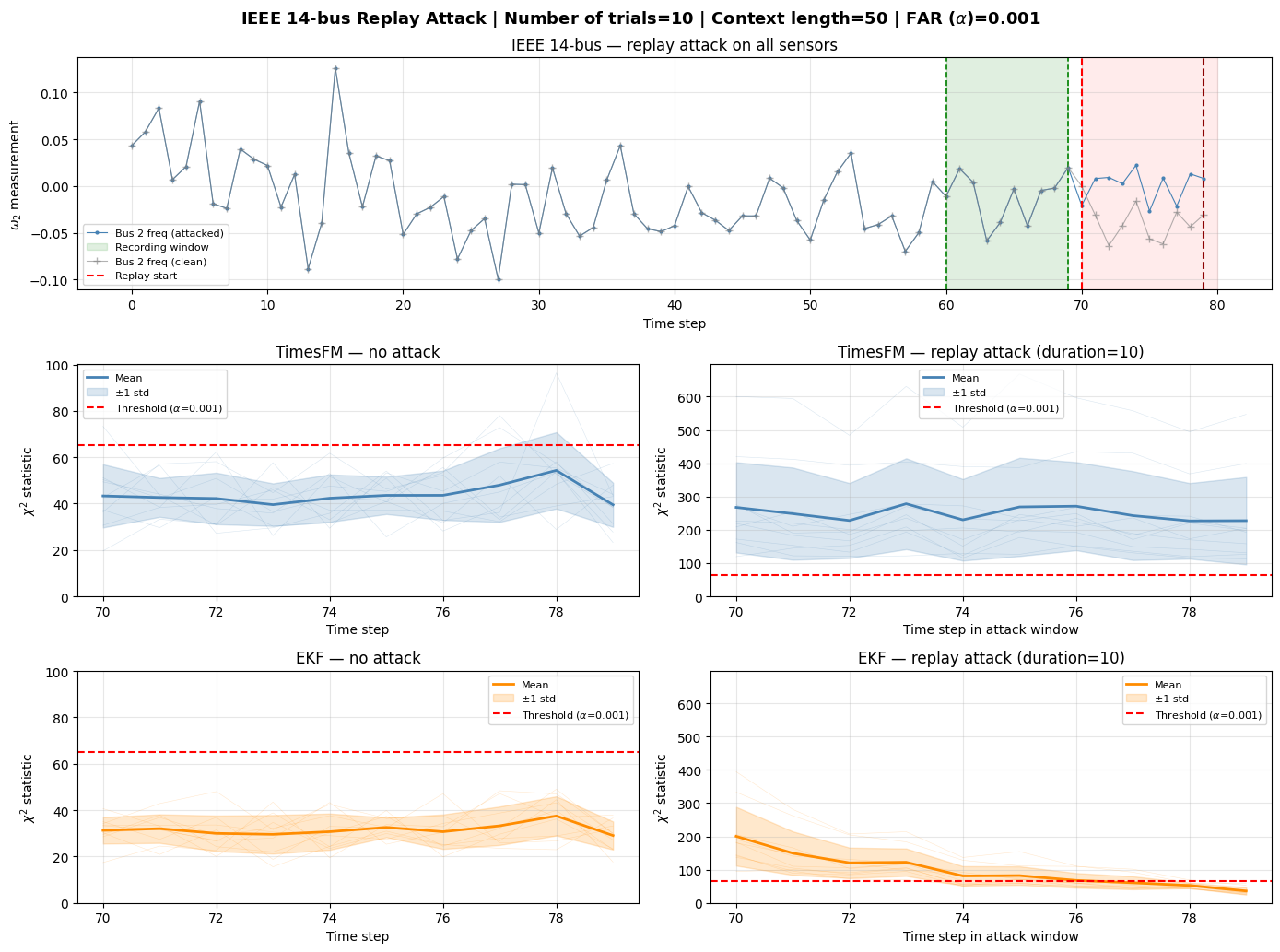}\\
    \vspace{4pt}
    \includegraphics[width=8cm]{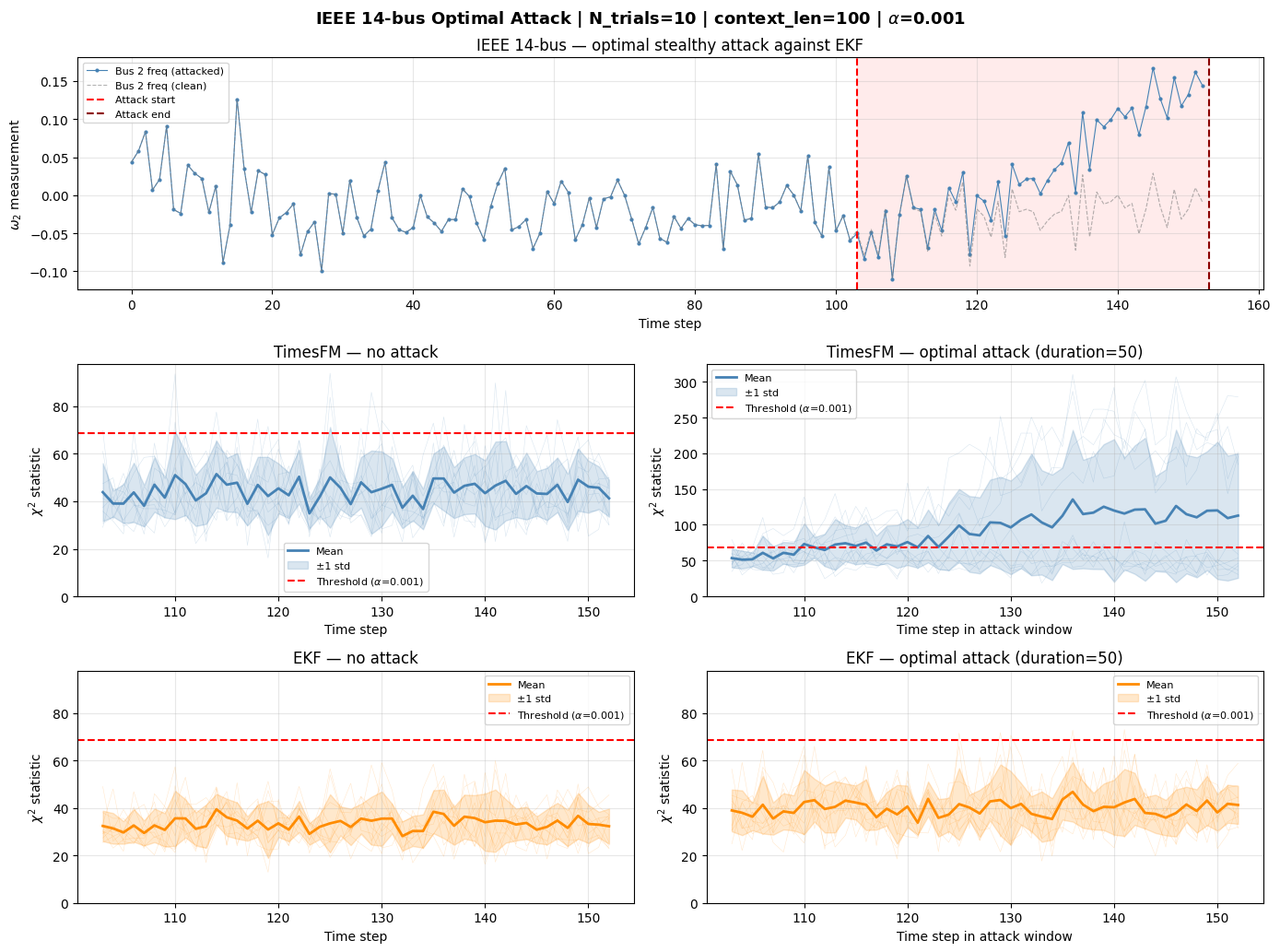}
    \caption{\small Evaluation of Algorithm~\ref{alg:LTI:detection} on the IEEE 14-bus system under replay attacks (top three rows) and stealthy attacks constructed using Theorem~\ref{thm:impact:NL} (bottom three rows). For each attack type: (Top) bus frequency $\omega_2[k]$ during nominal operation and attack phase; (Middle) mean and standard deviation of the test statistic across $10$ trials alongside threshold $\tau_p$, for the \tfm-based detector under nominal operation (left) and under attack (right); (Bottom) same for the EKF-based detector.}
    \label{fig:NL:results}
\end{figure}
\section{Conclusions}\label{sec:conclusion}
In this paper, we proposed a \tfm-based, model-structure-free attack detector. We derived closed-form optimal attack policies against a $\chi^2$ detector for both LTI and nonlinear systems, and showed that such attacks, while stealthy by construction against the model-based detector, can be reliably detected by the proposed \emph{TimesFM}-based secondary detector. Finally, we showed that \emph{TimesFM} predictions can serve as a mitigation tool, enabling robust state estimation even when all sensors are under attack. An important direction for future work is a rigorous characterization of adversaries with black-box query access to $g_s$; a preliminary empirical investigation is provided in Section~\ref{sec:adaptive}.
\bibliographystyle{ieeetr}
\bibliography{ref}
\end{document}